\documentclass[12pt]{amsart}
\usepackage{fullpage,hyperref}

\title{Factor-balanced $S$-adic languages}

\author[L.~Poirier]{L\'eo Poirier}
\address{Universit\'e de Lyon, ENS de Lyon, D\'epartement informatique de
l'ENS de Lyon, 46 all\'ee d'Italie, F-69007 Lyon, France}
\email{leo.poirier@ens-lyon.fr}
\author[W.~Steiner]{Wolfgang Steiner}
\address{Université Paris Cit\'e, CNRS, IRIF, F-75006 Paris, France}
\email{steiner@irif.fr}

\date{\today}
\thanks{This work was supported by the Agence Nationale de la Recherche through the project CODYS (ANR-18-CE40-0007).}

\newtheorem{lemma}{Lemma}[section]
\newtheorem{theorem}[lemma]{Theorem}

\newtheorem{proposition}[lemma]{Proposition}
\newtheorem{corollary}[lemma]{Corollary}

\numberwithin{equation}{section}

\newcommand{\cF}{\mathcal{F}} 
\newcommand{\cM}{\mathcal{M}} 
\newcommand{\cL}{\mathcal{L}} 
\newcommand{\bsigma}{\boldsymbol{\sigma}}
\newcommand{\bl}{\boldsymbol{\ell}}
\newcommand{\bv}{\mathbf{v}}

\begin{document} 
\begin{abstract}
A set of words, also called a language, is letter-balanced if the number of occurrences of each letter only depends on the length of the word, up to a constant. 
Similarly, a language is factor-balanced if the difference of the number of occurrences of any given factor in words of the same length is bounded.
The most prominent example of a letter-balanced but not factor-balanced language is given by the Thue--Morse sequence. 
We establish connections between the two notions, in particular for languages given by substitutions and, more generally, by sequences of substitutions. 
We show that the two notions essentially coincide when the sequence of substitutions is proper. 
For the example of Thue--Morse--Sturmian languages, we give a full characterisation of factor-balancedness.
\end{abstract}

\maketitle

\section{Introduction}
The study of balancedness of languages goes back at least to Morse and Hedlund~\cite{Morse-Hedlund40} who proved that each block of length $n$ in a Sturmian sequence of slope $\alpha$ has $\lfloor n\alpha \rfloor$ or $\lceil n\alpha \rceil$ occurrences of the letter that has frequency $\alpha$ (and thus $\lfloor n(1{-}\alpha)\rfloor$ or $\lceil n(1{-}\alpha)\rceil$ occurrences of the other letter).
In other words, the difference between the number of occurrences of a letter in blocks of the same length is at most 1; we call this property letter-1-balanced, previously it has often been simply called balanced.
More generally, a language is letter-balanced if the number of occurrences of a letter only depends on the length of a word in the language, up to an additive constant.
We do not only consider the occurrence of letters but also of longer blocks, and we say that a language is factor-balanced if the number of occurrences of each block in a word of the language only depends on the length of the word, up to a constant that can depend on the block. 
Usually, languages coming from a symbolic dynamical system (or subshift) are considered; we use the slightly weaker property of being factorial. 
For infinite sequences, balancedness is equivalent to bounded symbolic discrepancy, as studied in \cite{Adamczewski04}.
These concepts have applications in operations research, for optimal routing and scheduling and are related to Fraenkel's conjecture; see \cite{Berthe-Cecchi19} for references. 

Some relations between letter-balancedness and factor-balancedness have been studied in \cite{Adamczewski03,Queffelec10}, and more recently in \cite{Berthe-Cecchi19}.
Here, we improve on these results and show that factor-balancedness is preserved by the application of a substitution. 
We consider languages (or subshifts) given by a sequence of substitutions, also called $S$-adic languages.
When the substitutions are left or right proper, i.e., the image of each letter starts or ends with the same letter, we show that letter-balancedness (on all levels) implies factor-balancedness; this was previously known only under the assumption that the substitutions have unimodular incidence matrices \cite{BCDLPP21}. 
A particular case is that of a substitutive shift with a proper substitution. 
Here, we cannot remove the assumption of properness, as for example the Thue--Morse shift is not factor-balanced \cite{Berthe-Cecchi19}; we give a short proof in Section~\ref{sec:thue-morse-sturmian}.
For non-proper substitutions, we have to require balancedness for all factors of length $2$, not only for letters, in order to get factor-balancedness.

In Section~\ref{sec:balancedness}, we define most of the notions and give a characterisation of letter-balancedness in terms of the distance to a frequency vector. 
The effect of substitutions on balancedness is studied in Section~\ref{sec:substitutions}.
Section~\ref{sec:s-adic-languages} contains our main results, on sequences of substitutions. 
Finally, we consider the balancedness of a particular class of $S$-adic languages in Section~\ref{sec:thue-morse-sturmian}.

\section{Balancedness} \label{sec:balancedness}
For a finite alphabet~$A$, let $A^*$ be the set of finite words over~$A$.
A~\emph{language} is a subset $\cL \subseteq A^*$. 
A~word $v \in A^*$ is a \emph{factor} of $w \in A^*$ if there exist $p,s \in A^*$ such that $w = pvs$; here, $p$ is a \emph{prefix} and $s$ is a \emph{suffix} of~$w$.
We denote the set of factors of $w$ by $\cF(w)$, and the set of factors of elements of $\cL$ by $\cF(\cL)$. 
A~langage~$\cL$ is \emph{factorial} if $\cF(\cL) = \cL$. 
The length of a word $w \in A^*$ is denoted by $|w|$, i.e., $|w| = n$ if $w \in A^n$.
We denote the prefix (resp.\ suffix) of length~$n$ of a word~$w$ by $\mathrm{pref}_n(w)$ (resp.\ $\mathrm{suff}_n(w)$). 
The number of \emph{occurrences} of a word~$v$ in $w$, i.e., the number of different decompositions $w = pvs$, is denoted by $|w|_v$. 
A~language~$\cL$ is called \emph{$C$-balanced w.r.t.~$v$} if 
\[
\big| |w|_v-|w'|_v\big| \le C \qquad \mbox{for all $w,w' \in \cL$ with $|w| = |w'|$}.
\]
It is called \emph{$C$-balanced for length~$n$} if it is $C$-balanced for all $v \in A^n$.
We often omit the~$C$ and say that a language is \emph{balanced for length~$n$} if it is $C$-balanced for length~$n$ for some $C \ge 0$. 
Instead of ``($C$-)balanced for length~1'', we also say \emph{letter-($C$-)balanced}; other papers use the term ``balanced'' instead of letter-1-balanced or instead of letter-balanced, and a letter-balanced language is sometimes called ``finitely balanced''.
A~language is \mbox{\emph{factor-($C$-)balanced}} if it is ($C$-)balanced for all lengths $n\ge 1$. 
Note that factor-$C$-balancedness is a strong property that is satisfied for certain Sturmian languages~\cite{Fagnot-Vuillon02}, but we do not study this property here.
We are only interested in factor-balancedness, which means that for each~$n$ there exists $C_n$ such that~$\cL$ is $C_n$-balanced for length~$n$; equivalently, for each $v \in \cF(\cL)$ there exists $C_v$ such that $\cL$ is $C_v$-balanced w.r.t.~$v$. (Note that $|w|_v = 0$ for all $w \in \cL$, $v \notin \cF(\cL)$.) 
We first show that balancedness for length~$n$ is the same as balancedness for lengths up to~$n$.

\begin{lemma} \label{l:nto1}
If a language is balanced for length~$n$, then it is balanced for all lengths $k \le n$.
\end{lemma}

\begin{proof}
Let $\cL \subset A^*$ be $C$-balanced for length~$n$, $k < n$.
For all $v \in A^k$, $w \in A^*$, we have $|w|_v = \sum_{s\in A^{n-k}} |w|_{vs} + |\mathrm{suff}_{n-1}(w)|_v$, thus $\big||w|_v - |w'|_v\big| \le (\# A)^{n-k} C + n-1$ for $w' \in A^{|w|}$.
\end{proof}

We will also use the following characterisation of letter-balancedness in terms of distance from the line defined by a frequency vector, cf.\ \cite{Berthe-Tijdeman02,Adamczewski03}.

\begin{proposition} \label{p:frequency}
Let $\cL \subset A^*$ be an infinite letter-$C$-balanced factorial language. 
Then there exists a (frequency) vector $(f_a)_{a\in A}$ such that $||w|_a - f_a|w|| \le C$ for all $a \in A$, $w \in \cL$. 
\end{proposition}

\begin{proof}
Since $[0,1]^{\#A}$ is compact, there exists a vector $(f_a)_{a\in A}$ and a sequence of words $v_n \in \cL$ such that $\lim_{n\to\infty} |v_n| = \infty$ and $\lim_{n\to\infty} |v_n|_a/|v_n| = f_a$ for all $a \in A$. 
For arbitrary but fixed $w \in \cL$, set $k_n = \lfloor|v_n|/|w|\rfloor$, and decompose $v_n = v_{n,1} \cdots v_{n,k_n} v_{n,k_n+1}$ with $|v_{n,i}| = |w|$ for all $1 \le i \le k_n$.
Since $v_{n,i} \in \cL$, $\cL$~is letter-$C$-balanced, and $|v_{n,k_n+1}| < |w|$, we obtain that
\[
\bigg||w|_a - \frac{|v_n|_a}{|v_n|}\,|w|\bigg| = \frac{|w|}{|v_n|}\, \bigg||v_n|_a - \frac{|v_n|}{|w|}\, |w|_a\bigg| < \frac{|w|}{|v_n|} \bigg(|w| + \sum_{i=1}^{k_n} \big||v_{n,i}|_a - |w|_a\big|\bigg) \le \frac{|w|^2}{|v_n|} + \frac{k_n|w|}{|v_n|} C
\]
for all $a \in A$.
Letting $n \to \infty$, this gives that $||w|_a - f_a|w|| \le C$.
\end{proof}

\begin{lemma} \label{l:frequency2}
Let $\cL \subset A^*$ and $(f_a)_{a\in A}$ such that $||w|_a - f_a|w|| \le C$ for all $a \in A$, $w \in \cL$. 
Then $\cL$ is letter-$(2C)$-balanced.
\end{lemma}

\begin{proof}
We have $||w|_a|-|w'|_a| \le ||w|_a - f_a|w|| + |f_a|w| - f_a|w'|| \le 2C$ for all $w, w' \in \cL$, $a \in A$, such that $|w|=|w'|$.
\end{proof}

\section{Substitutions} \label{sec:substitutions}
In this section, we study how the application of a substitution influences balancedness. 
Here, a~\emph{substitution}~$\sigma$ is a morphism from $A^*$ to $B^*$, with the operation of concatenation, i.e., $\sigma(v w) = \sigma(v) \sigma(w)$ for all $v, w \in A^*$. 
We use the notation 
\[
\|\sigma\| := \max_{a\in A} |\sigma(a)|, \qquad \langle\sigma\rangle := \min_{a\in A} |\sigma(a)|,
\]
and call a substitution \emph{non-erasing} if all images of letters are non-empty, i.e., $\langle \sigma \rangle \ge 1$. 
It is left (resp.\ right) \emph{proper} when all letter images start (resp.\ end) with the same letter.

To show that substitutions preserve balancedness, we use the following lemma. 

\begin{lemma} \label{l:wxyz}
Let $\cL \subset A^*$ be a letter-$C$-balanced factorial language and $\sigma: A^* \to B^*$ a substitution.
Then, for all $w, w' \in \cF(\sigma(\cL))$ with $|w| = |w'|$, there exist $x, x', z, z' \in B^*$, $y, y' \in \cL$, such~that 
\[
\begin{gathered}
w = x\, \sigma(y)\, z, \ w' = x'\, \sigma(y')\, z', \quad |y| = |y'|, \ 
\max\{|x\, z|, |x' z'|\} \le (2+C\, \#A)\, \|\sigma\| - 2.
\end{gathered}
\]
\end{lemma}

\begin{proof}
Since $w, w' \in \cF(\sigma(\cL))$ and $\cL$ is factorial, we can write $w = x\, \sigma(v)\, u$ and $w = x' \sigma(v')\, u'$ with $v, v' \in \cL$, $u, u', x, x' \in A^*$ such that $|u|, |u'|, |x|, |x'| < \|\sigma\|$. 
Assume w.l.o.g.\ that $|v| \le |v'|$, let $y = v$, $v' = y' s'$ with $|y'| = |y|$, $z = u$, $z' = \sigma(s')\, u'$. 
Since $\cL$ is factorial, we have $y, y' \in \cL$.
Since $\cL$ is letter-$C$-balanced, we have 
\[
\big||\sigma(y)| - |\sigma(y')|\big| \le \sum_{a\in A} |\sigma(a)|\, \big||y|_a - |y'|_a\big| \le (\#A)\, C\, \|\sigma\|,
\]
thus
\[
|x' z'| = |w'| - |\sigma(y')| \le |w| - |\sigma(y)| + (\#A)\, C\, \|\sigma\| \le 2\, (\|\sigma\|-1) + (\#A)\, C\, \|\sigma\|.
\]
Therefore, $w = x\, \sigma(y)\, z$ and $w' = x'\, \sigma(y')\, z'$ satisfy all the required properties. 
\end{proof}

\begin{proposition} \label{p:sigma}
Let $\cL \subset A^*$ be a factorial language and $\sigma: A^* \to B^*$ a substitution.
If $\cL$ is letter-balanced, then $\cF(\sigma(\cL))$ is letter-balanced.
\end{proposition}

\begin{proof}
Suppose that $\cL$ is $C$-letter-balanced, and let $w = x\, \sigma(y)\, z$, $w' = x'\, \sigma(y')\, z'$ be as in Lemma~\ref{l:wxyz}.
Then, for all $b \in B$,
\begin{multline*}
\big||w|_b-|w'|_b\big| \le \big||xz|_b-|x'z'|_b\big| + \big||\sigma(y)|_b-|\sigma(y')|_b\big| \\
\le (2 + C\,\#A)\, \|\sigma\| - 2 + \sum_{a\in A} \big||y|_a-|y'|_a\big|\, |\sigma(a)|_b \le 2\,(1+C\,\#A)\, \|\sigma\| - 2. \qedhere
\end{multline*}
\end{proof}

To study balancedness for length~$n$, we use the \emph{$n$-coding} of a word $a_1 a_2 \cdots a_N \in A^N$, $N \ge 0$, which is the word over the alphabet~$A^n$ defined by
\[
(a_1 a_2 \cdots a_N)^{(n)} = (a_1 a_2 \cdots a_n) (a_2 a_3 \cdots a_{n+1}) \cdots (a_{N-n+1} a_{N-n+2} \cdots a_N) \in (A^n)^{N-n+1}
\]
if $N \ge n$, the empty word if $N < n$. 
Note that $|w|_v = |w^{(n)}|_v$ for all $v \in A^n$; in particular, a language $\cL$ is $C$-balanced for length~$n$ if and only if $\{w^{(n)} : w \in \cL\}$ is letter-$C$-balanced (over the alphabet~$A^n$). 

\begin{proposition} \label{p:sigmamn}
Let $\cL \subset A^*$ be a factorial language that is balanced for length~$n$, \mbox{$\sigma\!:\! A^*\! \to \!B^*$} a substitution, and $u \in B^*$ a (possibly empty) word that is a prefix of $\sigma(a) u$ for all $a \in A$ or a suffix of $u \sigma(a)$ for all~$a \in A$.
Then $\cF(\sigma(\cL))$ is balanced for length $\min_{w \in \cL \cap A^{n-1}} |\sigma(w)| {+} |u| {+} 1$.
In particular, 
\begin{itemize} 
\item
$\cF(\sigma(\cL))$ is balanced for length~$n$ if $\sigma$ is non-erasing, 
\item
$\cF(\sigma(\cL))$ is balanced for length $n{+}1$ if $\sigma$ is left or right proper. \end{itemize}
\end{proposition}

\begin{proof}
Let $\cL \subset A^*$, $\sigma: A^* \to B^*$, $u \in B^*$ be as in the statement of the lemma, $1 \le m \le \min_{w \in \cL \cap A^{n-1}} |\sigma(w)| {+} |u| {+} 1$.
Assume w.l.o.g.\ that $u$ is a prefix of $\sigma(a) u$ for all $a \in A$, the suffix case being symmetric. 
We define a substitution $\hat{\sigma}: (A^n \cap \cL)^* \to (B^m)^*$ by setting
\[
\hat{\sigma}(a_1 a_2 \cdots a_n) := \big(\sigma(a_1) \mathrm{pref}_{m-1}(\sigma(a_2 \cdots a_n)u)\big)^{(m)} \quad \mbox{for}\ a_1 \cdots a_n \in A^n \cap \cL.
\]
(Here, the alphabet is $A^n \cap \cL$). 
Then we have, for all $w \in \cL$,
\begin{equation} \label{e:sigmamn}
\big(\sigma(w) u\big)^{(m)} = \hat{\sigma}\big(w^{(n)}\big)\, \big(\sigma(\mathrm{suff}_{n-1}(w)) u\big)^{(m)}.
\end{equation}

Let $C_n$ be such that $\cL$ is $C_n$-balanced for length~$n$.
By Lemma~\ref{l:nto1}, $\cL$ is letter-$C_1$-balanced for some $C_1 \ge 0$.
Let $w, w' \in \cF(\sigma(\cL))$ with $|w| = |w'|$, and write $w = x\, \sigma(y)\, z$, $w' = x'\, \sigma(y')\, z'$ as in Lemma~\ref{l:wxyz}.
By~\eqref{e:sigmamn}, we have
\[
(wu)^{(m)} = \big(x\,\mathrm{pref}_{m-1}(\sigma(y)u)\big)^{(m)}\, \hat{\sigma}\big(y^{(n)}\big)\, \big(\sigma(\mathrm{suff}_{n-1}(y))z)u\big)^{(m)}.
\]
Using a similar decomposition for $(w'u)^{(m)}$, we obtain for $v \in B^m$ that
\[
\begin{aligned}
\big||w|_v - |w'|_v\big| & \le \max\big\{|xz|, |x'z'|\big\} + (n-1) \|\sigma\| + \sum_{t\in A^n \cap \cL} \big||y|_t -  |y'|_t\big|\, \big|\hat{\sigma}(t)\big|_v \\
& \le ((\#A)\, C_1 + 2) \|\sigma\| - 2 + (n-1) \|\sigma\| + (\#A)^{n-1}\, C_n\, \|\sigma\|.
\end{aligned}
\]
Here, we have used that $|w|_v = |w^{(m)}|_v$, that $\cL$ is $C_n$-balanced for length~$n$, and that $|\hat{\sigma}(a_1\cdots a_n)| = |\sigma(a_1)|$ for $a_1 \cdots a_n \in A^n \cap \cL$. 
This proves that $\cF(\sigma(\cL))$ is balanced for length $\min_{w \in \cL \cap A^{n-1}} |\sigma(w)| {+} |u| {+} 1$.
If $\sigma$ is non-erasing, then $|\sigma(w)| \ge n-1$ for all $w \in A^{n-1}$, thus $\cF(\sigma(\cL))$ is balanced for length~$n$.
If $\sigma$ is left or right proper, then $\sigma$ is non-erasing and $|u| \ge 1$, thus $\cF(\sigma(\cL))$ is balanced for length~$n{+}1$.
\end{proof}

\begin{theorem} \label{t:sigmafactor}
Let $\cL \subset A^*$ be a factorial language and $\sigma: A^* \to B^*$ a substitution.
If $\cL$ is factor-balanced, then $\cF(\sigma(\cL))$ is factor-balanced.
\end{theorem}

\begin{proof}
For non-erasing~$\sigma$, the theorem is a direct consequence of Proposition~\ref{p:sigmamn}. 
If $\cF(\sigma(\cL))$ is finite, then it is also factor-balanced. 
If $\cF(\sigma(\cL))$ is infinite, then there exists $a \in A$ such that $|\sigma(a)| \ge 1$ and $\{|w|_a : w \in \cL\}$ is unbounded.
If $\cL$ is letter-$C$-balanced, then Proposition~\ref{p:frequency} gives some $f_a \ge 0$ such that $|w|_a \ge f_a |w| - C$ and thus $|\sigma(w)| \ge f_a |w| - C$ for all $w \in \cL$.
By Proposition~\ref{p:sigmamn}, balancedness of~$\cL$ for length~$n$ implies balancedness of $\cF(\sigma(\cL))$ for length $f_a\, (n{-}1) {-} C {+} 1$.
Note that $f_a > 0$ since $|w|_a \le f_a |w| + C$ and $|w|_a$ is unbounded.
Therefore, factor-balancedness of~$\cL$ implies that of $\cF(\sigma(\cL))$.
\end{proof}

Sometimes it is also possible to infer letter-balancedness of~$\cL$ from that of~$\cF(\sigma(\cL))$. 
Here, the \emph{incidence matrix} of a substitution $\sigma: A^* \to B^*$ is 
\[
\cM_\sigma := (|\sigma(a)|_b)_{b\in B,a\in A}.
\]

\begin{proposition} \label{p:invertible}
Let $\cL \subset A^*$ be a factorial language, $\sigma: A^* \to B^*$ a substitution with invertible incidence matrix~$\cM_\sigma$. 
If $\cF(\sigma(\cL))$ is letter-balanced, then $\cL$ is \mbox{letter-balanced}.
\end{proposition}

\begin{proof}
If $\cF(\sigma(\cL))$ is letter-$C$-balanced, then Proposition~\ref{p:frequency} gives a vector $(f_b)_{b \in B}$ such that $||\sigma(w)|_b-f_b\,|\sigma(w)|| \le C$ for all $b \in B$, $w \in \cL$.
Since $\cM_\sigma$ is invertible, we can set $(f'_a)_{a\in A} := \cM_\sigma^{-1} (f_b)_{b \in B}$.
As $|\sigma(w)|_b = \sum_{a\in A} |w|_a\, |\sigma(a)|_b$, we obtain that $\big|\sum_{a\in A} |\sigma(a)|_b\, (|w|_a- |w| f'_a)\big| \le C$ for all $b \in B$, $w \in \cL$.
Then the invertibility of~$\cM_\sigma$ implies that there exists $C'$ such that $||w|_a - f'_a|w|| \le C'$ for all $w \in \cL$, $a \in A$.
Hence, by Lemma~\ref{l:frequency2}, $\cL$~is letter-$(2C')$-balanced.
\end{proof}

\section{$S$-adic languages} \label{sec:s-adic-languages}
Now, we consider sequences of substitutions $\bsigma = (\sigma_k)_{k\ge0}$, $\sigma_k: A_{k+1}^* \to A_k^*$.
We set
\[
\sigma_{[k,n)} := \sigma_k \circ \sigma_{k+1} \circ \cdots \circ \sigma_{n-1}, 
\]
for $n \ge k \ge 0$; then $\sigma_{[k,n)}$ is a substitution from $A_n^*$ to~$A_k^*$. 
The \emph{language of $\bsigma$ at level~$k$} is defined by 
\[
\cL_{\bsigma}^{(k)} := \big\{w \in A_k^* \,:\, w \in \cF(\sigma_{[k,n)}(A_n))\ \mbox{for infinitely many $n > k$}\big\},
\]
and $\cL_{\bsigma} := \cL_{\bsigma}^{(0)}$.
In other papers, the requirement for infinitely many $n > k$ is replaced by ``some $n > k$''; this can change the language only if a letter of~$A_m$ does not occur in~$\sigma_m$. 
Our definition ensures that
\[
\cF\big(\sigma_{[k,n)}(\cL_{\bsigma}^{(n)})\big) = \cL_{\bsigma}^{(k)} \quad \mbox{for all}\ n \ge k \ge 0.
\]

A~sequence of substitutions $(\sigma_k)_{k\ge0}$ is \emph{everywhere growing} if $\lim_{k\to\infty} \langle\sigma_{[0,k)}\rangle = \infty$. 
It is left (resp.\ right) \emph{proper} when for each $k \ge 0$ there exists $n > k$ such that $\sigma_{[k,n)}$ is left (resp.\ right) proper.
The following theorem was proved in \cite[Corollary~5.5]{BCDLPP21} for unimodular incidence matrices, i.e., $|\!\det \cM_{\sigma_k}| = 1$ for all $k \ge 0$. 

\begin{theorem} \label{t:main}
Let $\bsigma$ be a left or right proper sequence of substitutions.
If $\cL_{\bsigma}^{(k)}$ is letter-balanced for infinitely many~$k$, then $\cL_{\bsigma}$ is factor-balanced. 
\end{theorem}

\begin{proof}
Assume that $\cL_{\bsigma}^{(k)}$ is letter-balanced for infinitely many~$k$, which implies that it is letter-balanced for all~$k$ by Proposition~\ref{p:sigma}.
Since $\bsigma$ is left or right proper, there exist $0 = k_0 < k_1 < k_2 < \cdots$ such that $\sigma_{[k_i,k_{i+1})}$ is left or right proper for all $i \ge 0$. 
Therefore, by Proposition~\ref{p:sigmamn}, $\cL_{\bsigma}$ is balanced for all lengths $n \ge 1$.
\end{proof}

The following corollary is the particular case of Theorem~\ref{t:main} with constant sequence $\bsigma = (\sigma, \sigma, \dots)$ for some substitution $\sigma: A^* \to A^*$; we write $\sigma^\infty$ for $(\sigma, \sigma, \dots)$.
The \emph{language of a substitution} is $\cL_\sigma := \cL_{\sigma^\infty}$ (and consists of those $w \in A^*$ that are in $\cF(\sigma^n(A))$ infinitely~often).

\begin{corollary} \label{c:main}
Let $\sigma: A^* \to A^*$ be a substitution such that $\sigma^k$ is left or right proper for some $k \ge 1$. 
If $\cL_\sigma$ is letter-balanced, then $\cL_\sigma$ is factor-balanced. 
\end{corollary}

For invertible incidence matrices, letter-balancedness at level~$0$ implies letter-balancedness at all levels by Proposition~\ref{p:invertible}, which gives the following corollary of Theorem~\ref{t:main}.

\begin{corollary} \label{c:invertible}
Let $\bsigma = (\sigma_k)_{k\ge0}$ be a left or right proper sequence of substitutions with invertible incidence matrix $\cM_{\sigma_k}$ for all $k \ge 0$.
If $\cL_{\bsigma}$ is letter-balanced, then $\cL_{\bsigma}$ is \mbox{factor-balanced}. 
\end{corollary}

If $\bsigma$ is not proper, then we need balancedness for length~$2$ to infer factor-balancedness.

\begin{theorem} \label{t:main2}
Let $\bsigma$ be an everywhere growing sequence of substitutions such that $\cL_{\bsigma}^{(k)}$ is balanced for length~$2$ for infinitely many~$k$. 
Then $\cL_{\bsigma}$ is factor-balanced. 
\end{theorem}

\begin{proof}
By Proposition~\ref{p:sigmamn}, balancedness of $\cL_{\bsigma}^{(k)}$ for length~$2$ implies balancedness of $\cL_{\bsigma}$ for length $\langle \sigma_{[0,k)} \rangle + 1$. 
Since $\bsigma$ is everywhere growing, this implies that $\cL_{\bsigma}$ is factor-balanced. 
\end{proof}

For primitive substitutions $\sigma$, a sufficient condition for balancedness for length~$n$ of~$\cL_\sigma$ is given in \cite[Theorem~22]{Adamczewski03}, and it indicates that balancedness for length~$2$ and factor-balancedness are closely related; see also \cite[Section~5.4.3]{Queffelec10}.
We prove that balancedness for length~$2$ implies factor-balancedness for most substitutions.

\begin{corollary}
Let $\sigma: A^* \to A^*$ be an everywhere growing substitution.
If $\cL_\sigma$ is balanced for length~$2$, then $\cL_\sigma$ is factor-balanced. 
\end{corollary}

We remark that, in Theorem~\ref{t:main}, Corollary~\ref{c:invertible} and Theorem~\ref{t:main2}, factor-balancedness holds not only for~$\cL_{\bsigma}$ but for all $\cL_{\bsigma}^{(k)}$, $k \ge 0$.

\section{Thue--Morse--Sturmian languages} \label{sec:thue-morse-sturmian}

We conclude the paper by studying a special class of sequences of substitutions that occurs naturally in \cite{Steiner20,Komornik-Steiner-Zou}.
A~language $\cL_{\bsigma}$, $\bsigma \in \{L,M,R\}^\infty$, with substitutions
\[
\begin{aligned}
L:\ & 0 \mapsto 0, & \quad M:\ & 0 \mapsto 01, & \quad R:\ & 0 \mapsto 01, \\
& 1 \mapsto 10, & & 1 \mapsto 10, & & 1 \mapsto 1,
\end{aligned}
\]
is \emph{Thue--Morse--Sturmian} if $\bsigma$ is primitive.
Recall that a sequence of substitutions $(\sigma_n)_{n\ge0}$ is \emph{primitive} if, for each $k \ge 0$, there exists $n > k$ such that $|\sigma_{[k,n)}(a)|_b \ge 1$ for all $a \in A_n$, $b \in B_k$; in the case of $\bsigma \in \{L,M,R\}^\infty$, this means that $\bsigma$ does not end with the constant sequence $L^\infty$ or $R^\infty$. 
However, the following results also hold for non-primitive sequences.

\begin{proposition} \label{p:TMSbalanced}
For all $\bsigma \in \{L,M,R\}^\infty$, $\cL_{\bsigma}$~is letter-2-balanced.
\end{proposition}

\begin{proof}
For $\bsigma = (\sigma_k)_{k\ge0} \in \{L,M,R\}^\infty \setminus \{L,R\}^\infty$, let $n \ge 1$ be minimal such that $\sigma_n = M$.
We claim that $\cF(\sigma_{[0,n)}(\{01,10\}^*))$ is letter-2-balanced. 
Indeed, we have $\sigma_{[0,n)}(01) = 01w$ and $\sigma_{[0,n)}(10) = 10w$ for some $w \in \{0,1\}^*$. (This property is trivial for $\sigma_{[n,n)}$ and can be shown inductively for~$\sigma_{[k,n)}$, $0 \le k < n$, since $\sigma_k \in \{L,R\}$.) 
Let $v, v' \in \cF(\sigma_{[0,n)}(\{01,10\}^*))$ with $|v| = |v'|$. 
If $|v| \ge |w| {+} 3$ then we can write $v = pus$, $v'=p'u's'$ such that $|u|_0 = |u'|_0 = |w|_0{+}1$, $|u|_1 = |u'|_1 = |w|_1{+}1$ and $ps, p's' \in \cF(\sigma_{[0,n)}(\{01,10\}^*))$.
Since $|ps| = |p's'|$ and $|ps|_0 - |p's'|_0 = |v|_0 - |v'|_0$, it is sufficient to consider $|v| \le |w|{+}2$.
If $|v| \le |w|{+}1$, then $v, v' \in \cF(\sigma_{[0,n)}(\{01\}^*))$, and it is well known that this language is letter-1-balanced; see \cite[Chapter~2]{Lothaire02}.
For $|v| = |w|{+}2$, we have $|v|_0 = |w|_0{+}1$ or $v \in \{0w0, 1w1\}$. 
Therefore, $\cF(\sigma_{[0,n)}(\{01,10\}^*))$ and thus $\cL_{\bsigma}$ are letter-2-balanced.

Let now $\bsigma \in \{L,R\}^\infty$.
If $\bsigma$ contains infinitely many $L$'s and $R$'s, then $\cL_{\bsigma}$ is Sturmian and thus letter-2-balanced; see e.g.\ \cite[Chapter~2]{Lothaire02}.
Since $L^n(0) \,{=}\, 0$, $L^n(1) \,{=}\, 10^n$, $R^n(0) \,{=}\, 01^n$, $R^n(1) \,{=}\, 1$, for all $n \,{\ge}\, 0$, the languages $\cL_{L^\infty}$ and~$\cL_{R^\infty}$ are also letter-1-balanced.
Finally, if $\sigma_n = L$, $\sigma_k = R$ for all $k > n$, or $\sigma_n = R$, $\sigma_k = L$ for all $k > n$, then $\cL_{\bsigma} = \cF(\sigma_{[0,n)}(\{01\}^*))$, which is again letter-1-balanced.
\end{proof}

For a characterisation of factor-balancedness of Thue--Morse--Sturmian languages, we need the following lemma in order to show that applying any composition of substitutions $L, M, R$ to the Thue--Morse language, which is not factor-balanced, does not create a factor-balanced language.  
More precisely, we show for $\sigma \in \{L,M,R\}^*$ that $\sigma(011)$ occurs only trivially in $\sigma(w)$, $w \in \cL_{M^\infty}$.
Here, $\{L,M,R\}^*$ is the set of compositions of substitutions in $L, M, R$. 

\begin{lemma} \label{l:rec}
Let $\sigma \in \{L,M,R\}^*$, $a,b \in \{0,1\}$, $p,s,v \in \{0,1\}^*$, $k \ge 1$, such that 
\[
\sigma(avb) = p\, \sigma(01^k)\, s, 
\]
$p$ is a strict prefix of $\sigma(a)$ and $s$ is a strict suffix of~$\sigma(b)$.
Then $avb = 01^k$ (and $p,s$ are~empty) or $avb = 1^{k+1}$ (and $s$ is empty).  
\end{lemma}

\begin{proof}
The statement is clearly true when $\sigma$ is the identity. 
For $\sigma = \sigma_0 \circ \sigma_1 \circ \cdots \circ \sigma_n$, $\sigma_i \in \{L,M,R\}$, we prove the statement by induction on~$n$. 

Let first $\sigma_n = L$.
Then we have 
\[
\sigma_{[0,n)}(L(avb)) = p\, \sigma_{[0,n)}(0(10)^n)\, s \quad \mbox{and} \quad 11 \notin \cF(L(avb)).
\]
If $p$ is a strict prefix of $\sigma_{[0,n)}(a)$, in particular if $a=0$, then $\sigma_{[0,n)}(av'b') = p\, \sigma_{[0,n)}(01)\, s'$ for a prefix $av'b'$ of $L(avb)$ and a strict suffix $s'$ of~$\sigma_{[0,n)}(b')$. 
By the induction hypothesis and since $11 \not\in \cF(L(av'b'))$, this implies that $p$ is empty. 
Since $\sigma_{[0,n)}(0)$ starts with~$0$ and $\sigma_{[0,n)}(1)$ starts with~$1$, we obtain that $L(avb) = 0(10)^n$ and thus $avb = 01^n$. 
If $a=1$ and $p = \sigma_{[0,n)}(1)\, p'$, then $\sigma_{[0,n)}(0v'b') = p' \sigma_{[0,n)}(01)\, s'$ for a prefix $10v'b'$ of $L(1vb)$ and a strict suffix $s'$ of~$\sigma_{[0,n)}(b')$.  
Now, the induction hypothesis implies that $p'$ is empty, thus $L(avb) = (10)^{n+1}$, i.e., $avb = 1^{n+1}$. 

Let now $\sigma_n = M$. 
Then we have 
\[
\sigma_{[0,n)}(M(avb)) = p\, \sigma_{[0,n)}(01(10)^n)\, s \quad \mbox{and} \quad 111 \notin \cF(M(avb)).
\]
If $p$ is a strict prefix of $\sigma_{[0,n)}(a)$, then $\sigma_{[0,n)}(av'b') = p\, \sigma_{[0,n)}(011)\, s'$, hence $p$ is empty, and $avb = 01^n$. 
If $a=1$ and $p = \sigma_{[0,n)}(1)\, p'$, then we obtain that $avb = 1^{n+1}$. 
If $a=0$ and $p = \sigma_{[0,n)}(0)\, p'$, then $\sigma_{[0,n)}(1v'b') = p' \sigma_{[0,n)}(011)\, s'$, hence $p'$ is empty, which contradicts that $\sigma_{[0,n)}(1)$ and $\sigma_{[0,n)}(0)$ start with different letters. 

Finally, let $\sigma_n = R$. 
Then we have 
\[
\sigma_{[0,n)}(R(avb)) = p\, \sigma_{[0,n)}(01^{n+1})\, s \quad \mbox{and} \quad 1^{n+2} \notin \cF(R(avb)).
\]
If $p$ is a strict prefix of $\sigma_{[0,n)}(a)$, then $R(avb) = 01^{n+1}$, thus $avb = 01^n$. 
Otherwise, we have $a=0$ and $p = \sigma_{[0,n)}(0)\, p'$, thus $\sigma_{[0,n)}(1v'b') = p' \sigma_{[0,n)}(01^{n+1})\, s'$, hence $p'$ is empty, which contradicts that $\sigma_{[0,n)}(1)$ and $\sigma_{[0,n)}(0)$ start with different letters. 
\end{proof}

\begin{theorem} \label{t:TMSfactors}
Let $\bsigma = (\sigma_k)_{k\ge0} \in \{L,M,R\}^\infty$. 
Then $\cL_{\bsigma}$ is factor-balanced if and only if $\sigma_k \ne M$ for infinitely many~$k$.
\end{theorem}

\begin{proof}
By Proposition~\ref{p:TMSbalanced}, $\cL_{\bsigma}^{(k)}$ is letter-balanced for all $k \ge 0$. 
If $\bsigma$ does not end with~$M^\infty$, then it is right proper. 
If $\bsigma$ also does not end with~$L^\infty$ or~$R^\infty$, then it is everywhere growing, and we can apply Theorem~\ref{t:main}. 
If $\bsigma$ ends with $LR^\infty$ or $RL^\infty$, then we have seen in the proof of Proposition~\ref{p:TMSbalanced} that $\cL_{\bsigma} = \cF(\sigma(\{01\}^*))$ for some $\sigma \in \{L,M,R\}^*$, which is factor-balanced.
The case of $\cL_{L^\infty}$ and $\cL_{R^\infty}$ is similar.

Consider now $\bsigma$ ending with~$M^\infty$.
We first give a short direct proof that the Thue--Morse language~$\cL_{M^\infty}$ is not balanced for length~$2$; a more general proof was given in \cite{Berthe-Cecchi19}.
To this end, define recursively words $w_n, w_n' \in \cL_{M^\infty}$ of length $(4^n{+}2)/3$ by
\[
\begin{aligned}
w_1 & = 00, & M^2(w_{2n-1}) & = 0\,w_{2n}\,0, & M^2(w_{2n}) & = 1\,w_{2n+1}\,1, \\
w'_1 & = 01, & M^2(w'_{2n-1}) & = w'_{2n}\,01, & M^2(w'_{2n}) & = w'_{2n+1}\,10.
\end{aligned}
\]
Then we have $(w_1)^{(2)} = (00)$, $(w'_1)^{(2)} = (01)$,
\[
\begin{aligned}
M_2^2((w_{2n-1})^{(2)})\, (01)(11) & = (01)\,(w_{2n})^{(2)}, & M_2^2((w_{2n})^{(2)})\,(10)(00) & = (10)\,(w_{2n+1})^{(2)}, \\
M_2^2((w'_{2n-1})^{(2)})\, (10) & = (w'_{2n})^{(2)}, & M_2^2((w'_{2n})^{(2)})\, (01) & = (w'_{2n+1})^{(2)},
\end{aligned}
\]
with the substitution
\[
M_2:\, (\{0,1\}^2)^* \to (\{0,1\}^2)^*, \quad 
\begin{array}{ll}(00) \mapsto (01) (10), & (01) \mapsto (01) (11), \\ (10) \mapsto (10) (00), & (11) \mapsto (10) (01).\end{array}
\]
Using the abelianizations
\[
\bl_2(w) = \begin{pmatrix}|w|_{00} \\ |w|_{01} \\ |w|_{10} \\ |w|_{11}\end{pmatrix}, \quad \cM_{M_2} = \big(|M_2(cd)|_{ab})_{ab,cd\in\{00,01,10,11\}} = \begin{pmatrix}0&0&1&0 \\ 1&1&0&1 \\ 1&0&1&1 \\ 0&1&0&0\end{pmatrix},
\]
we obtain that
\[
\begin{aligned}
\bl_2(w_{2n}) & = \cM_{M_2}^2\, \bl_2(w_{2n-1}) + \bl_2(11), & \bl_2(w_{2n+1}) & = \cM_{M_2}^2\, \bl_2(w_{2n}) + \bl_2(00), \\
\bl_2(w'_{2n}) & = \cM_{M_2}^2\, \bl_2(w'_{2n-1}) + \bl_2(10), & \bl_2(w'_{2n+1}) & = \cM_{M_2}^2\, \bl_2(w'_{2n}) + \bl_2(00).
\end{aligned}
\]
The right eigenvectors of $\cM_{M_2}$ (to the eigenvalues $2,-1,0,1$) are
\[
\bv_2 = \begin{pmatrix}1\\2\\2\\1\end{pmatrix},\ \bv_{-1} = \begin{pmatrix}1\\{-}1\\{-}1\\1\end{pmatrix},\ \bv_0 = \begin{pmatrix}1\\0\\0\\{-}1\end{pmatrix},\ \bv_1 = \begin{pmatrix}1\\{-}1\\1\\{-}1\end{pmatrix},
\]
thus
\[
\bl_2(w_{2n}) = \frac{4^{2n}-1}{18} \bv_2 + \frac{2n}{3} \bv_{-1} - \frac{1}{2} \bv_0 \quad \mbox{and} \quad \bl_2(w'_{2n}) = \frac{4^{2n}-1}{18} \bv_2 - \frac{2n}{6} \bv_{-1} - \frac{1}{2} \bv_0,
\]
hence $\bl_2(w_{2n}) - \bl_2(w'_{2n}) = n\, \bv_{-1}$, i.e., 
\[
|w_{2n}|_{00}-|w'_{2n}|_{00} = |w'_{2n}|_{01}-|w_{2n}|_{01} = |w'_{2n}|_{10}-|w_{2n}|_{10} = |w_{2n}|_{11}-|w'_{2n}|_{11} = n.
\]

To finish the proof of the theorem, we have to show that $\cF(\sigma(\cL_{M^\infty}))$ is not factor-balanced for all $\sigma \in \{L,M,R\}^*$. 
Since $111 \notin \cL_{M^\infty}$, Lemma~\ref{l:rec} implies that $|\sigma(w)|_{\sigma(011)} = |w|_{011}$ for all $w \in \cL_{M^\infty}$, and we clearly have $0 \le |w|_{11} - |w|_{011} \le 1$. 
Therefore, we have 
\[
\big||\sigma(w)|_{\sigma(011)}  - |\sigma(w')|_{\sigma(011)}\big| \ge \big||w|_{11} - |w'|_{11}\big| - 1
\]
for all $w, w' \in \cL_{M^\infty}$.
Since $|w|_{11} - |w'|_{11}$ is unbounded for $w, w' \in \cL_{M^\infty}$ with $|w| = |w'|$, it is also unbounded when we restrict to $w, w'$ with $|w|_0 = |w'|_0$ (and $|w|_1 = |w|'_1$). 
Then we have $|\sigma(w)| = |\sigma(w')|$, thus $\sigma(\cL_{M^\infty})$ is not balanced for length $|\sigma(011)|$. 
\end{proof}

We remark that, by Proposition~\ref{p:sigmamn}, $\cF(\sigma \circ L(\cL_{M^\infty}))$ is balanced for length $|\sigma(0)| + 1$ for any substitution~$\sigma$.
On the other hand, we have seen in the proof of Theorem~\ref{t:TMSfactors} that $\cF(\sigma \circ L(\cL_{M^\infty}))$ is not balanced for length $|\sigma(01010)|$ for any substitution $\sigma \in \{L,M,R\}^*$.

\bibliographystyle{amsalpha}
\bibliography{factorbalance}

\newcommand{\etalchar}[1]{$^{#1}$}
\providecommand{\bysame}{\leavevmode\hbox to3em{\hrulefill}\thinspace}
\providecommand{\MR}{\relax\ifhmode\unskip\space\fi MR }
\providecommand{\MRhref}[2]{%
  \href{http://www.ams.org/mathscinet-getitem?mr=#1}{#2}
}
\providecommand{\href}[2]{#2}
\begin{thebibliography}{BCBD{\etalchar{+}}21}

\bibitem[Ada03]{Adamczewski03}
B.~Adamczewski, \emph{Balances for fixed points of primitive substitutions},
  Theoret. Comput. Sci. \textbf{307} (2003), no.~1, 47--75.

\bibitem[Ada04]{Adamczewski04}
\bysame, \emph{Symbolic discrepancy and self-similar dynamics}, Ann. Inst.
  Fourier (Grenoble) \textbf{54} (2004), no.~7, 2201--2234 (2005).

\bibitem[BCB19]{Berthe-Cecchi19}
V.~Berth\'{e} and P.~Cecchi~Bernales, \emph{Balancedness and coboundaries in
  symbolic systems}, Theoret. Comput. Sci. \textbf{777} (2019), 93--110.

\bibitem[BCBD{\etalchar{+}}21]{BCDLPP21}
V.~Berth\'{e}, P.~Cecchi~Bernales, F.~Durand, J.~Leroy, D.~Perrin, and
  S.~Petite, \emph{On the dimension group of unimodular {$\mathcal{S}$}-adic
  subshifts}, Monatsh. Math. \textbf{194} (2021), no.~4, 687--717.

\bibitem[BT02]{Berthe-Tijdeman02}
V.~Berth\'{e} and R.~Tijdeman, \emph{Balance properties of multi-dimensional
  words}, Theoret. Comput. Sci. \textbf{273} (2002), no.~1-2, 197--224.

\bibitem[FV02]{Fagnot-Vuillon02}
I.~Fagnot and L.~Vuillon, \emph{Generalized balances in {S}turmian words},
  Discrete Appl. Math. \textbf{121} (2002), no.~1-3, 83--101.

\bibitem[KSZ22]{Komornik-Steiner-Zou}
V.~Komornik, W.~Steiner, and Y.~Zou, \emph{Unique double base expansions},
  2022, arXiv:2209.02373.

\bibitem[Lot02]{Lothaire02}
M.~Lothaire, \emph{Algebraic combinatorics on words}, Encyclopedia of
  Mathematics and its Applications, vol.~90, Cambridge University Press,
  Cambridge, 2002.

\bibitem[MH40]{Morse-Hedlund40}
M.~Morse and G.~A. Hedlund, \emph{Symbolic dynamics {II}. {S}turmian
  trajectories}, Amer. J. Math. \textbf{62} (1940), 1--42.

\bibitem[Que10]{Queffelec10}
M.~Queff\'{e}lec, \emph{Substitution dynamical systems---spectral analysis},
  second ed., Lecture Notes in Mathematics, vol. 1294, Springer-Verlag, Berlin,
  2010.

\bibitem[Ste20]{Steiner20}
W.~Steiner, \emph{Thue-{M}orse-{S}turmian words and critical bases for ternary
  alphabets}, Bull. Soc. Math. France \textbf{148} (2020), no.~4, 597--611.

\end{thebibliography}
\end{document}